\documentclass[12pt]{article}

\usepackage[margin=1in]{geometry}

\usepackage{setspace}
\usepackage[table]{xcolor}

\usepackage[
    backend=biber,
    style=apa,
    natbib,
    backref=true,
  ]{biblatex}
\usepackage{hyperref}
\hypersetup{
    colorlinks=true,
    linkcolor=[rgb]{0,0,.7},
    urlcolor=[rgb]{.7,0,.7},,
    citecolor=[rgb]{0,.7,0},
}

\usepackage{amssymb}  

\usepackage{amsmath,amsthm,mathtools}

\usepackage{thmtools}

\declaretheorem[numberwithin=section]{theorem}
\declaretheorem[sibling=theorem, style=definition]{definition}
\declaretheorem[sibling=theorem]{lemma}

\declaretheorem[sibling=theorem]{corollary}

\declaretheorem[sibling=theorem]{proposition}

\makeatletter
\patchcmd{\ALG@step}{\addtocounter{ALG@line}{1}}{\refstepcounter{ALG@line}}{}{}
\newcommand{\ALG@lineautorefname}{Line}
\makeatother

\usepackage{makecell}
\usepackage{booktabs}
\usepackage{multirow}

\usepackage{csquotes}  
\usepackage{colortbl}  

\usepackage[multiple]{footmisc}

\usepackage{float}
\usepackage{comment}
\usepackage{soul}
\usepackage{wrapfig}

\usepackage{graphicx}
\usepackage{dsfont}
\usepackage{tikz}
\usetikzlibrary{calc}
\usetikzlibrary{shapes.geometric}
\usepackage{paralist}

\usepackage{wrapfig}

\usepackage{float}
\usepackage{comment}

\usepackage{caption}
\usepackage[labelformat=simple]{subcaption}

\usepackage{rotating}

\usepackage{comment}

\usepackage{tikz}
\usetikzlibrary{automata, positioning, calc, fit, shapes.misc}
\usetikzlibrary{matrix}
\usetikzlibrary{positioning}
\usetikzlibrary{arrows, decorations.pathmorphing}
\usetikzlibrary{decorations.pathreplacing,calligraphy}

\newcommand{\Z}{\mathbb{Z}}

\makeatletter
\newcommand\Ps@textstyle[2]{\mathbb{P}_{#1}\left[{#2}\right]}
\newcommand\Es@textstyle[2]{\mathbb{E}_{#1}\left[{#2}\right]}
\newcommand\Ps[2]{%
  \mathchoice 
  {\underset{{#1}}{\mathbb{P}}\left[{#2}\right]}
  {\Ps@textstyle{#1}{#2}}
  {\Ps@textstyle{#1}{#2}}
  {\Ps@textstyle{#1}{#2}}
}
\newcommand\Es[2]{%
  \mathchoice 
  {\underset{{#1}}{\mathbb{E}}\left[{#2}\right]}
  {\Es@textstyle{#1}{#2}}{\Es@textstyle{#1}{#2}}{\Es@textstyle{#1}{#2}}
}
\makeatother

\makeatletter
\def\tagform@#1{\maketag@@@{\ignorespaces#1\unskip\@@italiccorr}}
\let\orgtheequation\theequation
\def\theequation{(\orgtheequation)}
\makeatother

\newcommand{\Aut}{\mathrm{Aut}}

\newcommand{\ind}{\mathds{1}}
\newcommand{\abs}[1]{\left|#1\right|}		

\newcommand{\set}[1]{\left\{ #1 \right\}}   

\newcommand{\brac}[1]{\left( #1 \right)}    
\newcommand{\sqbrac}[1]{\left[ #1 \right]}  
\newcommand{\mc}{\mathcal}
\newcommand{\ip}[1]{\langle #1 \rangle}  

\newcommand{\Description}[1]{}

\hypersetup{pdftitle={Existence of Fair Resolute Voting Rules}, pdfauthor={Manik Dhar, Kunal Mittal, Clayton Thomas}}

\begin{document}

\title{Existence of Fair Resolute Voting Rules}
\author{
  Manik Dhar\thanks{Massachusetts Institute of Techonology. {E-mail}: \href{mailto:dmanik@mit.edu}{dmanik@mit.edu}.}
  \and 
  Kunal Mittal\thanks{Courant Institute of Mathematical Sciences, New York University. {E-mail}: \href{mailto:kunal.mittal@nyu.edu}{kunal.mittal@nyu.edu}. Research supported by a Simons Investigator Award to Subhash Khot. Part of this research was completed while the author was at Princeton University, supported by NSF Award CCF-2007462 and a Simons Investigator Award to Ran Raz.}
  \and 
  Clayton Thomas\thanks{Rensselaer Polytechnic Institute. {E-mail}: \href{mailto:}{thomas.clay95@gmail.com}.}
}
\date{February 14, 2026}

\begin{titlepage}
\maketitle
\begin{abstract}
  Among two-candidate elections that treat the candidates symmetrically and never result in a tie, 
which voting rules are fair?
A natural requirement is that each voter exerts an equal influence over the outcome, i.e., is equally likely to swing the election one way or the other.
A voter's influence has been formalized in two canonical ways: the Shapley-Shubik (\citeyear{shapley1954method}) index and the Banzhaf (\citeyear{banzhaf1964weighted}) index.
We consider both indices, and ask: Which electorate sizes admit a fair voting rule (under the respective index)?

For an odd number $n$ of voters, simple majority rule is an example of a fair voting rule.
However, when $n$ is even, fair voting rules can be challenging to identify, and a diverse literature has studied this problem under different notions of fairness.
Our main results completely characterize which values of $n$ admit fair voting rules under the two canonical indices we consider.
For the Shapley-Shubik index, a fair voting rule exists for $n>1$ if and only if $n$ is not a power of $2$.
For the Banzhaf index, a fair voting rule exists for all $n$ except $2$, $4$, and $8$.
Along the way, we show how the Shapley-Shubik and Banzhaf indices relate to the winning coalitions of the voting rule, and compare these indices to previously considered notions of fairness.
\end{abstract}
\thispagestyle{empty}
\end{titlepage}

\maketitle

\section{Introduction}

In a two-candidate election with $n$ voters, a voting rule is \emph{resolute} if it never results in a tie. 
The classic Majority voting rule is resolute whenever $n$ is odd, and a Dictatorship (electing the choice of a single distinguished voter) is resolute for any $n$.
Conventional wisdom suggests that Majority is significantly more \emph{fair} than Dictatorship;
for instance, Majority satisfies the celebrated fairness property of \emph{anonymity}, meaning that the election depends only on the set of votes placed, regardless of who placed which vote.
 
While anonymity doubtlessly guarantees a high degree of fairness, it is also highly restrictive.
Indeed, anonymity precludes many voting rules that might reasonably be considered highly fair, such as the one presented in \autoref{fig:maj-of-maj}.
Moreover, May's theorem (\citeyear{May52}) implies that whenever $n$ is even, there are \emph{no} resolute, anonymous voting rules, precluding such fair voting rules  for these electorate sizes.

\begin{figure}[htbp]
  \begin{minipage}{0.45\linewidth}
    \centering
    \includegraphics[width=0.80\linewidth]{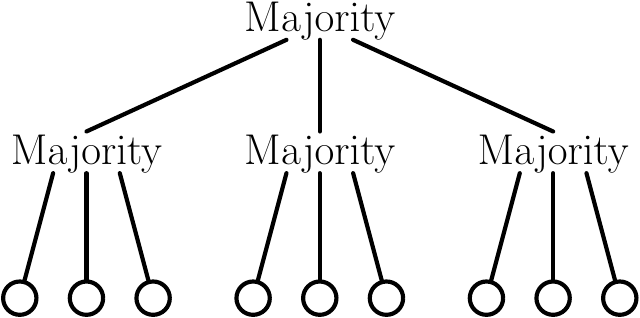}
  \end{minipage}
  \begin{minipage}{0.53\linewidth}
    \caption{A fair, non-anonymous voting rule}
    \label{fig:maj-of-maj}
    \Description{Majority of Majority, described in figure notes.}
    
    { \footnotesize \textbf{Notes:}
    Circles depict $n=9$ voters, who are arranged in $3$ groups of $3$. 
    The voting rule calculates the majority vote within each group, then elects the majority of these three majorities.
    Such voting rules are termed ``representative democracies'' in \citet{BartholdiHJTY21}, who observe that such rules are fair in the sense that every voter plays an identical, symmetric role in the election (see \autoref{sec:prelims-other-stuff}).
    \par }
  \end{minipage}
\end{figure}

In this paper, we consider more permissive notions of fairness, and ask: For which $n$ does there exist a resolute, fair voting rule with $n$ voters?\footnote{
  Throughout the paper, we assume that voting rules are \emph{monotone}, meaning that voting for a candidate can only support electing that candidate, and \emph{neutral}, meaning that the two candidates are treated identically.
}
Our fairness notions require that every voter exerts an equal amount of influence over the outcome. 
This degree of influence has been formalized in two canonical ways: The Shapley-Shubik (\citeyear{shapley1954method}) index and the Banzhaf (\citeyear{banzhaf1964weighted}) index.\footnote{
  The Shapley-Shubik index of voter $i$ is the probability---over all possible orderings of the voters---that $i$ is the first in the order to turn a losing coalition into a winning one.
  The Banzhaf index of voter $i$ is the probability---over all subsets of voters excluding $i$---that adding $i$ to the subset changes it from a losing to a winning coalition.
  See Section~\ref{sec:voting_rules_and_fair_indices}. 
}
These two well-studied indices satisfy distinct properties, and have been compared and contrasted across a number of studies \citep[e.g.,][]{DubeyS79,laruelle2001shapley}.
Nevertheless, the implications of enforcing fairness under these indices remain largely unknown.\footnote{
  A stronger notion of fairness, termed \emph{equitability}, has received comparatively more attention. See \autoref{sec:prelims}.
}

Whenever $n$ is odd, our question has an easy answer: Majority rule is resolute and fair.
For even $n$, resolute and fair voting rules may initially seem counterintuitive.
Indeed, for some small even numbers such as $n=2$, it is easy to show that fair rules do not exist.\footnote{
  For $n=2$, a complete argument goes as follows.
  Suppose voter $1$ votes for candidate $a$, and voter $2$ for $b$.
  Since the voting rule is resolute, some candidate must be elected; suppose without loss of generality it is candidate $a$. 
  Denote this by $f(a,b)=a$.
  Then, monotonicity implies that $f(a,a)=a$, while neutrality implies that $f(b,a)=b$ and $f(b,b)=b$.
  Hence, $f$ is a Dictatorship in which voter $1$ decides the outcome, and this rule is not fair (under any fairness notion we consider).
}
However, it turns out that fair resolute voting rules are possible with even $n$, and many examples of such rules are already known.
For a concrete example, see \autoref{fig:icos}.

\begin{figure}[htbp]
  \begin{minipage}{0.38\linewidth}
    \centering
    \includegraphics[width=0.9\linewidth]{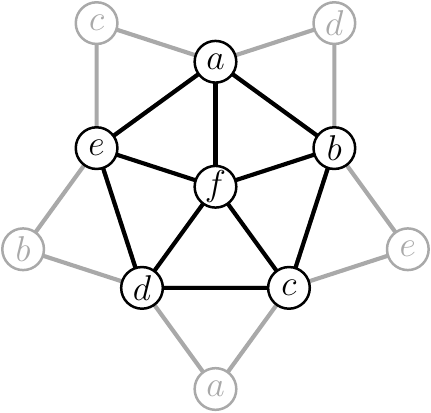}
  \end{minipage}
  \begin{minipage}{0.6\linewidth}
    \caption{A fair rule with an even number of voters}
    \label{fig:icos}
    \Description{Icos, described in figure notes.}
    
    { \footnotesize \textbf{Notes:}
    The $n=6$ voters, labeled $\{a,b,c,d,e,f\}$, are represented by either one or two nodes in a graph.
    Each of the $10$ triangles in this graph represents a triple of voters called a minimal winning coalition (MWC); these MWCs are pairwise intersecting. 
    One can show that for every profile of votes, at least one MWC always votes for the same candidate;
    this candidate is the outcome of the voting rule.
    This voting rule is called ``Icos'' in \citet{LoebC00},
    and one can show that this rule is fair in the same sense as the voting rule in \autoref{fig:maj-of-maj} (see \autoref{sec:prelims-other-stuff}).
    \par }
  \end{minipage}
\end{figure}

In our main results, we exactly characterize the $n$ for which there exist resolute, fair voting rules, under the above fairness notions.
First, for Shapley-Shubik fairness:

\begin{theorem}\label{thm:intro_ss}
There exists a Shapley-Shubik-fair resolute voting rule on $n$ voters if and only if $n\not=2, 4, 8, 16,32 \dots$.
\end{theorem}

Second, for Banzhaf fairness, our result is:
\begin{theorem}\label{thm:intro_banzhaf}
   There exists a Banzhaf-fair resolute voting rule on $n$ voters if and only if $n\not=2, 4, 8$.
\end{theorem}

As an intermediate step, we show how the fairness notions relate to the winning coalitions of the voting rule.
Specifically, we give a formula for the Shapley-Shubik index, and Banzhaf index, of voter $i$, in terms of (only) the number of winning coalitions of different sizes that include $i$ (and which does not reference, e.g., whether $i$ is pivotal for certain winning coalitions).
This also allows us to show that a different previously considered fairness notion, termed unbiasedness, is stronger than both Shapley-Shubik-fairness and Banzhaf-fairness.

Our results show that fairness under these indices is considerably more permissive than it might at first appear. 
For Shapley-Shubik fairness, the fraction of electorate sizes $n\le M$ for which no fair voting rule exists goes to zero as $M$ goes to infinity.
For Banzhaf fairness, even more is true: all but a constant number of electorate sizes admit a fair voting rule.
This may open a path for using fair and resolute voting rules in-practice. 
\section{Preliminaries}
\label{sec:prelims}

\subsection{Voting Rules and Fairness Indices}\label{sec:voting_rules_and_fair_indices}

Let $N = \{1,\ldots,n\}$ denote a set of $n$ voters.
A (resolute, two-candidate) \emph{voting rule} is a function $f : 2^N \to \{0,1\}$. 
Here, $0$ and $1$ represent the two candidates in the election.
If $S\subseteq N$ denotes all voters who vote for $1$, then $f(S)$ gives the result of the election.
We restrict attention to voting rules satisfying the following properties:
\begin{enumerate}
    \item (Monotone) Increasing the set of votes for a candidate can only increase their chances of being elected; i.e., $f(S)\leq f(T)$ for every $S\subseteq T\subseteq N$.
    \item (Neutral) The voting rule treats the candidates identically; i.e. $f(N\setminus S)=1-f(S)$ for each $S\subseteq N$.
\end{enumerate}

A set $S\subseteq N$ is called a \emph{winning coalition} of $f$ if $f(S)=1$.
Since $f$ is monotone, any superset of a winning coalition is also a winning coalition.
Since $f$ is neutral, a winning coalition $S$ can always elect either candidate of their choosing; i.e., if $S$ all vote for candidate $0$, the result is $f(N\setminus S)=0$.
Since $f$ is resolute, exactly one of $S$ or $N\setminus S$ is always a winning coalition for each $S\subseteq N$.
Moreover, any two winning coalitions are intersecting; formally:

\begin{lemma}[Winning Coalitions are Intersecting]\label{lemma:wc_intersect}
	Let $f:2^N\to \set{0,1}$ be a voting rule, and let $S,T\subseteq N$ be such that $f(S)=f(T)=1$.
	Then, $S\cap T \not= \emptyset$.
\end{lemma}
\begin{proof}
	Suppose that $S\cap T = \emptyset$.
	Then, $S\subseteq N\setminus T$, and using monotonicity and neutrality, we get
	$1 = f(S) \leq f(N\setminus T) = 1-f(T)=0,$
	which is a contradiction.
\end{proof}

For a concrete example, \autoref{fig:n-is-4-guy} shows an example of a voting defined in terms of its winning coalitions.
In the introduction, \autoref{fig:icos} defined a voting rule via a similar approach;\footnote{
  \autoref{fig:icos} refers to \emph{minimal} winning coalitions (MWCs), which are winning coalitions such that no strict subset is a winning coalition.
}
in \autoref{fig:maj-of-maj}, the winning coalitions are all subsets of voters that include, for at least two of the three subgroups, at least two of the three voters within the subgroup.

\begin{figure}[htbp]
  \begin{minipage}{0.4\linewidth}
    \centering
    \includegraphics[width=0.6\linewidth]{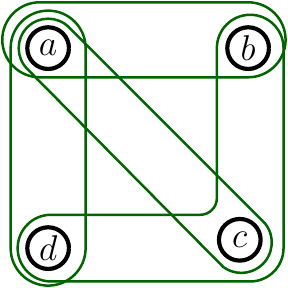}
  \end{minipage}
  \begin{minipage}{0.5\linewidth}
    \caption{A (non-fair) voting rule}
    \label{fig:n-is-4-guy}
    \Description{A non-fair voting rule, described in figure notes.}

    { \footnotesize \textbf{Notes:}
    The voters are labeled $N=\{a,b,c,d\}$.
    Winning coalitions are depicted as curved shapes enclosing two or three voters. 
    $N$ is also a winning coalition.
    In words, the election works as follows.
    If any of $\{b,c,d\}$ place the same vote as $a$, then that candidate is elected. 
    Otherwise, all voters in $\{b,c,d\}$ must place the same vote, and that candidate is elected.
    \par }
  \end{minipage}
\end{figure}

We study two canonical notions from the voting theory literature of the influence a voter exerts on the outcome of the voting rule.
First, the \emph{Shapley-Shubik index} $\varphi_i$ \citep{shapley1953value, shapley1954method} of voter $i$ is defined to be the probability, given a uniformly random ordering of the voters, that the set of voters before and including $i$ is winning but the set of voter strictly before $i$ is not winning.
Formally,
\begin{align*}
  \varphi_i 
  & = \Ps{\sigma \sim S_n}{ 
    f\bigl(\{j\in N: \sigma(j)\leq \sigma(i)\}\bigr)=1\ \land\
     f\bigl(\{j\in N: \sigma(j) < \sigma(i)\}\bigr)=0 }
  \\ & = \sum_{S\subseteq N, S \ni i}
    \frac{ (|S|-1)!\cdot  (n-|S|)!}{n!} \cdot
    \bigl( f(S) - f(S\setminus \set{i}) \bigr).
\end{align*}
Observe that these indices sum to 1, i.e. $\sum_{i\in N}\varphi_i = 1$.

Second, the \emph{Banzhaf index} $\beta_i$ \citep{penrose1946elementary, banzhaf1964weighted} of voter $i$ is defined to be the probability, given a uniformly random set $S$ of voters other than $i$, that when $i$ is added to $S$ the result is a winning coalition, but that $S$ by itself is not a winning coalition.
Formally,
\begin{align*}
  \beta_i 
  & = \Ps{S \sim 2^{N\setminus \set{i}}}{ 
    f(S \cup \set{i})=1\ \land\  f(S)=0 }
  \\ & = \frac{1}{2^{n-1}} 
  \sum_{S\subseteq N, S \ni i}
    \bigl( f(S) - f(S\setminus \set{i}) \bigr).
\end{align*}

For example, for the voting rule defined in \autoref{fig:n-is-4-guy}, we have
\begin{align*}
  \varphi_a=\frac{1}{2}, \qquad \varphi_b=\varphi_c=\varphi_d= \frac{1}{6}, \qquad\qquad\qquad \beta_a=\frac{3}{4}, \qquad \beta_b=\beta_c=\beta_d=\frac{1}{4}.
\end{align*}

We study voting rules that are fair in the sense that every voter has the same influence over the outcome, according to one of the above indices.
Formally, we call a voting rule \emph{Shapley-Shubik-fair} (resp. \emph{Banzhaf-fair}) if every voter has the same Shapley-Shubik index (resp. Banzhaf index).

\subsection{Related Notions}
\label{sec:prelims-other-stuff}

We also study voting rules that satisfy a stronger notion of fairness, termed \emph{unbiasedness}.

To define unbiasedness, let $\mu^{(p)}$ denote the distribution over subsets $S$ of voters that, independently for each voter $i\in N$, includes $i$ in $S$ with probability $p$ and excludes $i$ with probability $1-p$.
In other words, $\mu^{(p)}$ is the distribution over $2^N$ such that $\mu^{(p)}(S)=p^{|S|}(1-p)^{n-|S|}$ for each $S\in 2^N$. 
Then, define
\begin{align*}
  \beta_i^{(p)} 
  & = \Ps{S \sim \mu^{(p)}}{ 
    f(S \cup \set{i})=1\ \land\ f(S\setminus \{i\})=0 }.
\end{align*}
Roughly, this measures the probability that voter $i$ is pivotal, when the distribution over the voters is $p$-biased.\footnote{This coincides with the notion of of $p$-biased influence, a central and well-studied concept in the analysis of Boolean functions; see \cite{Don14} for an excellent introduction to this subject.}
We call $f$ \emph{unbiased} if, for all $p\in [0,1]$, we have $\beta_i^{(p)}=\beta_j^{(p)}$ for every pair of voters $i,j\in N$.

Any unbiased voting rule is also Shapley-Shubik-fair and Banzhaf-fair.
For Banzhaf, this follows immediately from the observation that $\beta_i=\beta_i^{(1/2)}$.
For Shapley-Shubik, we establish this fact by counting winning coalitions.
In fact, we prove the following proposition relating winning coalitions to unbiasedness, Shapley-Shubik-fairness and Banzhaf-fairness.

\begin{proposition}
\label{prop:winning_coalition_properties}
    Let $f:2^N\to \set{0,1}$ be a voting rule.
    For each $k\in \set{0,1,\dots,n}$, $i\in N$, let $w_i^{(k)}$ denote the number of winning coalitions of size $k$ containing the voter $i$, i.e., \[w_i^{(k)} = \abs{\set{S\subseteq N :  f(S)=1,\ S\ni i,\ \abs{S}=k}}.\]
    Then, the following hold:
    \begin{enumerate}
        \item The voting rule $f$ is unbiased if and only if for all voters $i,j\in N$ and all $k\in\{0,1,2,\dots,n\}$,  we have $w_i^{(k)}=w_j^{(k)}$.
        
        \item For each voter $i\in N$, the Shapley-Shubik index satisfies \[ \varphi_i = -1+2\cdot \sum_{k=1}^n \frac{ (k-1)! (n-k)!}{n!} \cdot w_i^{(k)} . \]

        \item For each voter $i\in N$, the Banzhaf index satisfies \[ \beta_i = -1 + \frac{1}{2^{n-2}} \sum_{k=1}^n w_i^{(k)}.\]
    \end{enumerate}
    In particular, if $f$ is unbiased, it is Shapley-Shubik-fair and Banzhaf-fair.
\end{proposition}

This result may be of independent interest.
We prove \autoref{prop:winning_coalition_properties} below in Section~\ref{sec:unbiased_implies_others}.

For an additional point of comparison, one can we consider \emph{equitable} voting rules, in the sense of \citet{BartholdiHJTY21}.\footnote{
  The problem of characterizing equitable (and resolute, neutral, and monotone) voting rules has been studied in-depth in group theory.
  A famous open problem known as Isbell's conjecture (\citeyear{Isbell60}) is to characterize which $n$ permit such voting rules with $n$ voters.
  Other relevant works include \citet{CameronFK89} and \citet{LoebC00}.
  We find Shapley-Shubik-fairness and Banzhaff-fairness comparatively easy to characterize, and (in contrast to this open problem for equitability) we completely characterize which $n$ permit such fair rules.
}
In contrast to our main fairness notions, which are based around the degree of influence a voter has, equitability is based around the idea of voters playing symmetric roles.
To define equitability, say that a permutation $\sigma : N \to N$ is a \emph{symmetry} of $f$ if $S\subseteq N$ is a winning coalition of $f$ if and only if $\sigma(S) = \{\sigma(i) \mid i \in S\}$ is a winning coalition of $f$.
Then, $f$ is equitable if the group of symmetries of $f$ is transitive, that is, if for all pairs of voters $i,j$, there is a symmetry $\sigma$ of $f$ such that $\sigma(i)=j$.
If $f$ is equitable, then it is not hard to see that it is unbiased (and hence, additionally Shapley-Shubik-fair and Banzhaf-fair).
Figures~\ref{fig:maj-of-maj} and~\ref{fig:icos} in the introduction provide examples of equitable voting rules.
\cite{Bhatnagar20} constructs examples of voting rules that are unbiased but not equitable; we also provide an example in \autoref{sec:size-9-unbiased}.

\autoref{fig:notion-relation} shows the relationship between all notions of fairness we consider.

\def\middleName{Unbiased}
\def\innerName{Equitable}

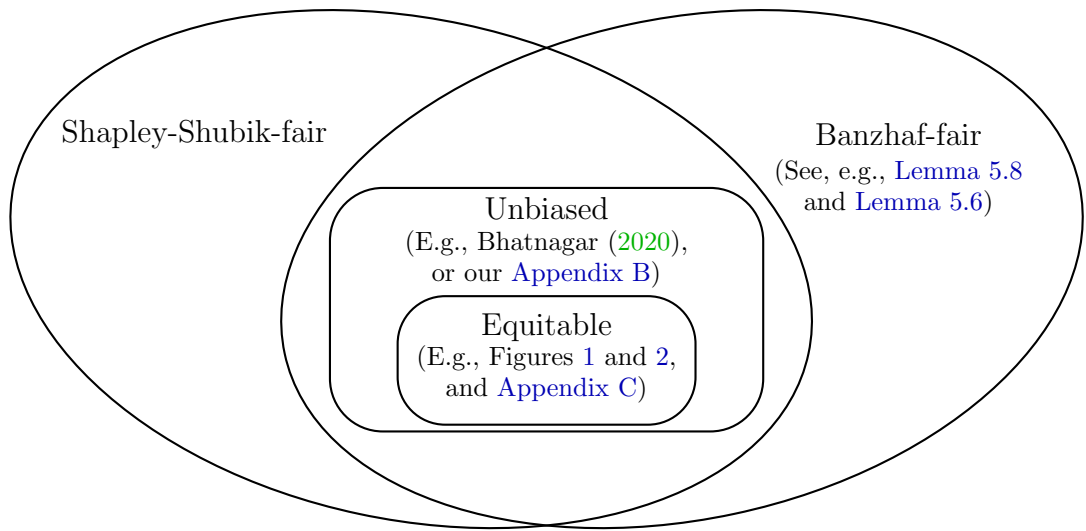
\begin{figure}[htbp]
  \Description{Relation between fairness notions, described in text.}
  {
  \centering
  \footnotesize
  \begin{tikzpicture}[x=1cm,y=1cm,
      every node/.style={font=\footnotesize},
      scale=0.9
  ]
  
  \draw[thick,rotate around={-12:(-2,0.4)}]
      (-2,0.4) ellipse (6 and 3.7);
  \node[anchor=east] at (-3.1,2.4) {\normalsize Shapley-Shubik-fair};
  
  \draw[thick,rotate around={12:(2,0.4)}]
      (2,0.4) ellipse (6 and 3.7);
  \node[anchor=center] at (5.2,2.4) {\normalsize Banzhaf-fair};
  \node[align=center] at (5.2,1.6) {(See, e.g., \autoref{lem:exists-banzhaf-fair} \\and \autoref{lem:shapley-non-existence})};

  \draw[thick,rounded corners=20pt]
      (-3.2,-2.0) rectangle (3.2,1.6);
  \node at (0,1.3) {\normalsize Unbiased};
  \node[align=center] at (0,0.55) {(E.g., \citet{Bhatnagar20},\\or our \autoref{sec:size-9-unbiased})};
  
  \draw[thick,rounded corners=18pt]
      (-2.2,-1.9) rectangle (2.2,0.0);
  \node at (0,-0.45) {\normalsize Equitable};
  \node[align=center] at (0,-1.15) {(E.g., Figures~\ref{fig:maj-of-maj} and~\ref{fig:icos},\\and \autoref{app:equit_ex})};
  
  \end{tikzpicture}
  }
  
  \caption{Relation between fairness notions}
  \label{fig:notion-relation}
\end{figure}

\section{Unbiased implies Shapley-Shubik-Fair and Banzhaf-Fair}\label{sec:unbiased_implies_others}

In this section, we prove each of the items in Proposition~\ref{prop:winning_coalition_properties}.
Let $f:2^N\to \set{0,1}$ be a voting rule.
For each voter $i\in N$, and $k\in \set{0,1,\dots,n}$, let $w_i^{(k)}$ denote the number of winning coalitions of size $k$ containing voter $i$, i.e., \[w_i^{(k)} = \abs{\set{S\subseteq N :  f(S)=1,\ S\ni i,\ \abs{S}=k}}.\]

First, we characterize unbiasedness in terms of the above winning coalition counts, as follows:

\begin{proposition}
\label{prop:unbiased_win_coal}
    Let $f:2^N\to \set{0,1}$ be a voting rule.
    For any voter $i\in N$, and integer $k$, define:
    \begin{itemize}
        \item The coalitions of size $k$ pivotal for voter $i$ by \[\mc A_i^{(k)} = \set{S\subseteq N: \abs{S}=k,\ f(S\cup\set{i})=1,\ f(S\setminus\set{i})=0}.\]
        \item The winning coalitions of size $k$ pivotal for voter $i$ by \[\mc B_{i}^{(k)} = \set{S\subseteq N \ : \ \abs{S}=k,\ S\ni i,\ f(S)=1,\ f(S\setminus\set{i}) = 0}.\]
        \item The winning coalitions of size $k$ containing voter $i$ by \[\mc W_{i}^{(k)} = \set{S\subseteq N \ : \ \abs{S}=k,\ S\ni i,\ f(S)=1}.\]
    \end{itemize}

	Then, the following are equivalent:
	\begin{enumerate}
		\item The voting rule $f$ is unbiased.
		\item For every $k\in \set{0,1,\dots,n}$ and voters $i,j\in N$, it holds that $|\mc A_i^{(k)}| = |\mc A_j^{(k)}|$.
		\item For every $k\in \set{0,1,\dots,n}$ and voters $i,j\in N$, it holds that $|\mc B_i^{(k)}| = |\mc B_j^{(k)}|$.
		\item For every $k\in \set{0,1,\dots,n}$ and voters $i,j\in N$, it holds that $|\mc W_i^{(k)}| = |\mc W_j^{(k)}|$.
	\end{enumerate}
\end{proposition}
\begin{proof}
	Let $f:2^N \to \set{0,1}$ be a voting rule.
	The equivalences are proved as follows.

    \paragraph{\boldmath $(1\iff 2)$:} For any voter $i\in N$ and $p\in [0,1]$, we can write
		\[ \beta_{i}^{(p)} = \Ps{S \sim \mu^{(p)}}{ 
    f(S \cup \set{i})=1\ \land\ f(S\setminus \{i\})=0 } =  \sum_{k=0}^n p^k (1-p)^{n-k}\cdot \abs{\mc A_{i,k}}.\]
		Observe that the polynomials $\set{p^k (1-p)^{n-k}:k=0,1,\dots, n}$ are linearly independent (over the reals).
		Hence, for any $i,j\in N$, it holds that $\beta_{i}^{(p)} = \beta_{j}^{(p)}$ for every $p\in [0,1]$ if any only if $|\mc A_i^{(k)}| = |\mc A_j^{(k)}|$ for every $k\in \set{0,1,\dots,n}$.

    \paragraph{\boldmath $(2\iff 3)$:} For any voter $i\in N$, and  $k\in \set{0,1,\dots,n}$, we have
    \[ |\mc A_{i}^{(k)}| = |\{S\in \mc A_{i}^{(k)}: S\ni i\}| + |\{S\in \mc A_{i}^{(k)}: S\not \ni i\}|
			= |\mc B_i^{(k)}| + |\mc B_i^{(k+1)}|.\]
   
    Now, suppose voters $i,j\in N$ are such that $|\mc B_i^{(k)}| = |\mc B_j^{(k)}|$ for every $k=0,1,\dots,n$.
    Then, by the above (along with the fact $|\mc B_i^{(n+1)}| = 0$ for every $i\in N$), it holds that $|\mc A_i^{(k)}| = |\mc A_j^{(k)}|$ for every $k=0,1,\dots,n$.
    
	Conversely, suppose voters $i,j\in N$ are such that $|\mc A_i^{(k)}| = |\mc A_j^{(k)}|$ for every $k=0,1,\dots,n$.
    Then, induction on $k$ (along with the fact $|\mc B_i^{(0)}| = 0$ for every $i\in N$) shows that $|\mc B_i^{(k)}| = |\mc B_j^{(k)}|$ for every $k=0,1,\dots,n$.
		
	\paragraph{\boldmath $(3\iff 4)$:} For any voter $i\in N$ and $k\in \set{0,1,\dots,n}$, we define \[\mc W_{\lnot i}^{(k)} = \set{S\subseteq N \ : \ \abs{S}=k,\ S\not\ni i,\ f(S)=1}.\]
    Observe that, for any voter $i\in N$, and $k\in \set{1,2,\dots, N}$,
    \[ \mc W_i^{(k)} = \mc B_i^{(k)} \cup \set{S\cup\set{i}: S\in \mc W_{\lnot i}^{(k-1)}}, \]
    and moreover, the two sets on the right hand side of this equality are disjoint.
    Hence, \[ |\mc W_i^{(k)}| = |\mc B_i^{(k)}| + |\mc W_{\lnot i}^{(k-1)}|. \]

    Now, suppose voters $i,j\in N$ are such that $|\mc W_{ i}^{(k)}| = |\mc W_{j}^{(k)}|$ for every $k=0,1,\dots,n$. 
    Then, we must have $|\mc W_{\lnot i}^{(k)}| = |\mc W_{\lnot j}^{(k)}|$ for every $k=0,1,\dots,n$ as well.   
    Hence, by the above (along with the fact $|\mc B_i^{(0)}| = 0$ for every $i\in N$), we must have $|\mc B_i^{(k)}| = |\mc B_j^{(k)}|$ for every $k=0,1,\dots,n$ as well.

	Conversely, suppose voters $i,j\in N$ are such that $|\mc B_{ i}^{(k)}| = |\mc B_{j}^{(k)}|$ for every $k=0,1,\dots,n$.
    Then, induction on $k$ (along with the fact $|\mc W_i^{(0)}| = |\mc W_{\lnot i}^{(0)}|= 0$ for every $i\in N$) shows that $|\mc W_i^{(k)}| = |\mc W_j^{(k)}|$ and $|\mc W_{\lnot i}^{(k)}| = |\mc W_{\lnot j}^{(k)}|$ for every $k=0,1,\dots,n$.
\end{proof}

\begin{corollary}\label{corr:unbiased_winning_coalition_count}
   The voting rule $f$ is unbiased if and only if for all voters $i,j\in N$ and all $k\in\{0,1,2,\ldots,n\}$, we have $w_i^{(k)}=w_j^{(k)}$.
\end{corollary}
\begin{proof}
    This follows from Proposition~\ref{prop:unbiased_win_coal}, by observing that $w_i^{(k)} = |\mc W_i^{(k)}|$ for every voter $i\in N$ and $k\in \set{0,1,\dots, N}$.
\end{proof}

Next, we characterize Shapley-Shubik-fair by the winning coalition counts, as follows:

\begin{proposition}
For each voter $i\in N$, the Shapley-Shubik index satisfies \[ \varphi_i = -1+2\cdot \sum_{k=1}^n \frac{ (k-1)! (n-k)!}{n!} \cdot w_i^{(k)} . \]
\end{proposition}
\begin{proof}
We have
\begin{align*}
     \varphi_i &= \sum_{S\subseteq N, S \ni i}
    \frac{ (|S|-1)!\cdot  (n-|S|)!}{n!} \cdot
    \brac{f(S) - f(S\setminus \set{i}}
    \\&= \sum_{k=1}^n \frac{ (k-1)! (n-k)!}{n!}  \sum_{S\subseteq N, S \ni i, \abs{S}=k} \brac{f(S) - f(S\setminus \set{i}}
\end{align*}
Since $f$ is neutral, this equals
\begin{align*}
    \varphi_i &= \sum_{k=1}^n \frac{ (k-1)! (n-k)!}{n!}  \sum_{S\subseteq N, S \ni i, \abs{S}=k} \brac{f(S) - 1 + f( (N\setminus S)\cup \set{i})}
    \\&= \sum_{k=1}^n \frac{ (k-1)! (n-k)!}{n!}  \cdot \brac{w_i^{(k)} -\binom{n-1}{k-1} + w_i^{(n-k+1)} }
    \\&= \sum_{k=1}^n \frac{ (k-1)! (n-k)!}{n!}  \cdot w_i^{(k)} -1 + \sum_{k=1}^n \frac{ (k-1)! (n-k)!}{n!}  \cdot w_i^{(n-k+1)}
    \\&= -1 + 2\cdot  \sum_{k=1}^n \frac{ (k-1)! (n-k)!}{n!}  \cdot w_i^{(k)} .\qedhere
\end{align*}
\end{proof}

The analogous result for Banzhaf index is:

\begin{proposition}
    For each voter $i\in N$, the Banzhaf index satisfies \[ \beta_i = -1 + \frac{1}{2^{n-2}} \sum_{k=1}^n w_i^{(k)}.\]
\end{proposition}
\begin{proof}
   Since $f$ is neutral, we have
    \begin{align*}
        \beta_i &= \frac{1}{2^{n-1}} \sum_{S\subseteq N, S \ni i} \brac{f(S) - f(S\setminus \set{i}}
        \\&= \frac{1}{2^{n-1}} \sum_{S\subseteq N, S \ni i} \brac{f(S) - 1 + f((N\setminus S)\cup \set{i}}
        \\&= -1 + \frac{1}{2^{n-1}}\sum_{k=1}^n \brac{w_i^{(k)} + w_i^{(n-k+1)}}
        \\&= -1 + \frac{1}{2^{n-2}}\cdot \sum_{k=1}^n  w_i^{(k)}.\qedhere
    \end{align*}
\end{proof}

\section{A Construction of Unbiased Voting Rules}

In this section, we prove the following theorem:

\begin{theorem}\label{thm:unbiased_even_nottwopow}
    Let $n\geq 1$ be an even integer that is not a power of 2.
    Then, there exists an unbiased voting rule on $n$ voters.
\end{theorem}

\subsection{Complementary Balanced Families}

We start by showing that for even $n$, the existence of unbiased (resp. equitable) voting rules is equivalent to the existence of a complementary balanced families (resp. complementary transitive families) of $\frac{n}{2}$-sized subsets, defined as follows:

\begin{definition}
	Let $n\geq 2$ be an even integer, and let $N=\set{1,2,\dots, n}$.
    For a family $\mc T\subseteq 2^N$ of  $\frac{n}{2}$-sized subsets, define the following properties:
	\begin{enumerate}
		\item (Complementary) For every $S\subseteq N$ of size $\frac{n}{2}$, exactly one of the conditions $S\in \mc T$ and $N\setminus S\in \mc T$ holds.
		\item (Balanced) For every $i\in [n]$, the set $\mc T_i = \set{S\in \mc T : S\ni i}$ is of the same size.
		\item (Transitive) The group $\Aut(\mc T) = \set{\sigma\in S_n : \sigma(S) \in \mc T \text{ for all }S\in \mc T}$ is transitive.
	\end{enumerate}
	We call $\mc T$ complementary balanced (resp. transitive) if it satisfies the complementary and balanced (resp. transitive) properties.
\end{definition}
Note that complementary transitive families are also complementary balanced.

\begin{lemma}\label{lemma:comp_bal_iff_unbiased}
	Let $n\geq 2$ be an even integer.
	Then, the following hold:
	\begin{enumerate}
		\item There exists an unbiased voting rule $f:2^N\to \set{0,1}$ if and only if there exists a complementary balanced family of $\frac{n}{2}$-sized subsets.
		\item There exists an equitable voting rule $f:2^N\to \set{0,1}$ if and only if there exists a complementary transitive family of $\frac{n}{2}$-sized subsets.
	\end{enumerate}
\end{lemma}
\begin{proof}
    Suppose there exists an unbiased (resp. equitable) voting rule $f:2^N\to \set{0,1}$.
    Then, the collection \[\mc T = \set{S\subseteq [n]:\ \abs{S}=\frac{n}{2},\ f(S)=1}\]
    is easily checked to be complementary balanced (resp. transitive); see Corollary~\ref{corr:unbiased_winning_coalition_count}.
    
    Conversely, suppose that there exists a complementary balanced (resp. transitive) family $\mc T\subseteq 2^N$ of $\frac{n}{2}$-sized subsets.
    Then, the voting rule $f:2^N\to \set{0,1}$ given by \[ f(S) = \begin{cases}
		1, & \abs{S}=\frac{n}{2},\ S\in \mc T \\ 1,&  \abs{S}>\frac{n}{2} \\ 0, &  \text{otherwise}
	\end{cases}\]
	is unbiased (resp. equitable); see Corollary~\ref{corr:unbiased_winning_coalition_count}.
\end{proof}


\subsection{Design Construction}

A special ingredient in our construction of unbiased voting rules is a \emph{combinatorial design}, which we construct as follows:

\begin{lemma}\label{lemma:design_construction} (Balanced Family of Measure Half)
	Let $k\geq 1$ be an integer such that $\binom{2k+1}{k}$ is even.
	Then, there exists a collection $\mc S$ of $k$-sized subsets of $\set{1,2,\dots,2k+1}$, such that:
    \begin{enumerate}
        \item $\abs{\mc S} = \frac{1}{2}\binom{2k+1}{k}$, and
        \item Each element $i\in \set{1,2,\dots,2k+1}$ occurs in an equal number of subsets in $\mc S$.
    \end{enumerate}
\end{lemma}
\begin{proof}
	Let $k\geq 1$ be an integer such that $\binom{2k+1}{k}$ is even.
    Let $G = \Z/(2k+1)\Z$, and let $\binom{G}{k} = \set{S\subseteq G : \abs{S}=k}$ denote the collection of all $k$-sized subsets of $G$.
	We shall exhibit a collection $\mc S\subseteq \binom{G}{k}$ satisfying the required properties.
    
	First, we show that for any $S\in \binom{G}{k}$, its $G$-orbit $G\cdot S = \set{S+i \pmod{2k+1} : i\in G}$ is of size $2k+1$.\footnote{We use $S+i \pmod{2k+1}$ to denote the set $\set{a+i\pmod{2k+1}: a\in S}$.} Equivalently, we may show that for any $i\in G\setminus\set{0}$, it holds that $S+i \pmod{2k+1} \not= S$.
    For this purpose, consider any $i\in G\setminus\set{0}$, let $\ip{i}\subseteq G$ denote the cyclic subgroup generated by $i$, and let $r = |\ip{i}| > 1$; note that $r$ must divide $|G|=2k+1$.
    Suppose, for the sake of contradiction, that $S+i \pmod{2k+1} = S$.
    Then, the action of $\ip{i}$ on $S$ partitions the set $S$ into $\ip{i}$-orbits, each of size $r$; in particular, $r$ divides $|S|=k$.
    Hence, $r$ divides $\gcd(k,2k+1)=1$, which is a contradiction.
    
	The above shows that $\binom{G}{k}$ is partitioned into $\frac{1}{2k+1} \binom{2k+1}{k}$ disjoint $G$-orbits, each of size $2k+1$.
	Since $\binom{2k+1}{k}$ is even and $2k+1$ is odd, $\frac{1}{2k+1} \binom{2k+1}{k}$ is also even.
	Now, we let $\mc S$ to be the union of any $\frac{1}{2}\cdot \frac{1}{2k+1} \binom{2k+1}{k}$ such orbits.
	Since $G$ is a transitive group, and $\mc S$ is $G$-invariant, this satisfies the required properties.
\end{proof}


\subsection{Complementary Balanced Families via Designs}

Using the design in the previous section, we construct a complementary balanced family.

\begin{lemma}\label{lemma:comp_balanced_exist}
	Let $n\geq 1$ be an even integer such that $\binom{n}{n/2}$ is divisible by 4.
	Then, there exists a complementary balanced family of $\frac{n}{2}$-sized subsets.
\end{lemma}

\begin{proof}
	Let $n\geq 1$ be an even integer such that $\binom{n}{n/2}$ is divisible by 4, and let $m = \frac{n}{2} \geq 3$.
	We show the existence of a complementary balanced family $\mc T$ of $\frac{n}{2}$-sized subsets.
	
	Since $\binom{2m}{m} = 2\cdot \binom{2m-1}{m-1}$, the lemma hypothesis implies that $\binom{2m-1}{m-1}$ is even.
    Hence, by Lemma~\ref{lemma:design_construction} (with $k=m-1$), we can find a family $\mc S$ of $(m-1)$-sized subsets of $\set{1,2,\dots,2m-1}$, of size $\abs{\mc S} = \frac{1}{2}\binom{2m-1}{m-1}$, such that each element $i\in \set{1,2,\dots,2m-1}$ occurs in an equal number of subsets in $\mc S$.
    Now, we define:
    \[ \mc U  = \set{S \cup \set{2m} : S\in \mc S}, \]
    \[\mc V  = \set{S: S\subseteq \set{1,2,\dots,2m-1},\ \abs{S}=m,\ \set{1,2,\dots,2m-1}\setminus S \not\in \mc S}.\]
    In words, $\mc U$ adds the element $2m$ to every element of $\mc S$, and $\mc V$ includes the complements (within $\set{1,2,\dots,2m-1}$) of sets not in $\mc S$.
	Observe that $\mc U$ and $\mc V$ are disjoint, as each set in $\mc U$ contains the element $2m$, and no set in $\mc V$ contains the element $2m$.
	We define our required family $\mc T$ to be $\mc T:=\mc U \cup \mc V$.
        
	\paragraph{Complementary:} We show that the collection $\mc T$ is complementary.\footnote{Indeed, this holds for any family $\mc T$ defined as above for \emph{any} collection $\mc S$ of $(m-1)$-sized subsets of $\set{1,2,\dots,2m-1}$.}
    Consider any set $T\subseteq \set{1,2,\dots,2m}$ of size $m$; we shall show that exactly one of $T$ and $\overline{T}=\set{1,2,\dots,2m}\setminus T$ lie in the family $\mc T$.
    Without loss of generality, we may assume that $2m\in T$ (else we can work with the complement $\overline{T}$ of $T$).
    Denoting $S = T\setminus \set{2m}$, we have:
	\begin{enumerate}
		\item Suppose $S\in \mc S$: In this case, $T \in \mc U \subseteq \mc T$ by definition.

        On the other hand, $\overline{T}\not\in \mc U$ since $2m\not\in\overline{T}$, and $\overline{T}\not\in \mc V$ since $\set{1,2,\dots,2m-1}\setminus \overline{T} = S \in \mc S$.
		Hence, $\overline{T}\not\in \mc T$.
		
		\item Suppose $S\not\in \mc S$.
		In this case $T\not\in \mc U$ since $S\not\in \mc S$, and $T\not\in \mc V$ since $2m\in T$.
		Hence, $T\not\in \mc T$.
		
        On the other hand, we have $\overline{T} = \set{1,2,\dots,2m-1}\setminus S $ is of size $m$, and satisfies that $\set{1,2,\dots,2m-1}\setminus \overline{T} = S\not\in \mc S$.
		Hence $\overline{T}\in \mc V\subseteq \mc T$.
	\end{enumerate}

    \paragraph{Balanced:} We show that for every $i\in \set{1,2,\dots,2m}$, the set $\mc T_i = \set{T\in \mc T: T\ni i}$ is of size $\frac{1}{4}\binom{2m}{m}$.
    In particular, this implies that $\mc T$ is balanced.
    
	For $i=2m$, this follows by observing that $\mc T_{2m} = \mc U$, and hence $\abs{\mc T_{2m}} = \abs{\mc S}=\frac{1}{4}\binom{2m}{m}$.
	Now, consider any index $i\in \set{1,2,\dots,2m-1}$, and let $\mc S_i = \set{S\in \mc S: S\ni i}$; note that this has size $|\mc S_i| = \frac{m-1}{2m-1}\cdot \frac{1}{2}\binom{2m-1}{m-1} = \frac{m-1}{4(2m-1)}\binom{2m}{m}$ using the fact that each $\mc S_i$ is of equal size.
	
    Observe that $\abs{\mc T_i \cap \mc U} = \abs{\mc S_i}$, and
    \begin{align*}
        \abs{\mc T_i\cap \mc V} &= \abs{\set{S: S\subseteq \set{1,2,\dots,2m-1},\ \abs{S}=m,\ \set{1,2,\dots,2m-1}\setminus S \not\in \mc S, i\in S}}
        \\&= \abs{\set{S: S\subseteq \set{1,2,\dots,2m-1},\ \abs{S}=m-1,\ S \not\in \mc S, i\not\in S}}
        \\&= \abs{\set{S: S\subseteq \set{1,2,\dots,2m-1},\ \abs{S}=m-1,\ i\not\in S}} \\&\qquad- \abs{\set{S: S\subseteq \set{1,2,\dots,2m-1},\ \abs{S}=m-1,\ S \in \mc S, i\not\in S}}
        \\&= \binom{2m-2}{m-1} - \brac{\abs{\mc S}-\abs{\mc S_i}}
        \\&= \binom{2m}{m}\cdot \sqbrac{\frac{m}{2(2m-1)} - \frac{1}{4} + \frac{m-1}{4(2m-1)}}
        \\&= \binom{2m}{m}\cdot \frac{m}{4(2m-1)}.
    \end{align*}
	Hence,\footnote{Note that the above logic would apply to any family $\mc T$ of winning coalitions of a voting rule with an even number $n=2m$ of voters, where we fix voter $2m$ and define $\mc U$ to be all the winning coalitions of size $m$ which include voter $2m$, and $\mc V$ to be all those that do not contain voter $2m$, and for each $i \in \set{1,2,\dots,2m-1}$ consider $\alpha_i = \abs{ \set{ U \in \mc U : i \in U }}$. The above logic proves that if voting rule is unbiased then $\alpha_i$ is independent of the choice of $i\in \set{1,2,\dots,2m-1}$; and that $\abs{ \set{ V \in \mc V : i \in V }} = \binom{2m-2}{m-1} - \abs{\mc U}+\alpha_i$ is also independent of the choice of $i$.} we get
	\[ \abs{\mc T_i} = \abs{\mc T_i\cap\mc U}+\abs{\mc T_i\cap\mc V} = \frac{1}{4}\binom{2m}{m}.\qedhere\]
\end{proof}

\subsection{Combining Things}

We characterize the integers for which the divisibility condition above holds.
\begin{lemma}\label{lemma:divisibility}
	Let $m\geq 1$ be an integer.
	Then $\binom{2m}{m}$ is divisible by 4 if and only if $m \not\in \set{2^r: r=0,1,2,\dots}$.
\end{lemma}
\begin{proof}
	Let $m\geq 1$ be an integer. Since $\binom{2m}{m} = 2\cdot \binom{2m-1}{m-1}$, the condition we care about is equivalent to $\binom{2m-1}{m-1}$ being divisible by 2.
	
	For $m=1$ the statement is true, so we assume $m>1$ henceforth.
	Suppose that the binary representation of $m-1$ is $a_{r-1}a_{r-2}\dots a_0$, for some $r\geq 1$, each $a_i\in \set{0,1}$, and $a_{r-1}=1$.
	Then, the binary representation of $2m-1 =2(m-1)+1$ is $a_{r-1}a_{r-2}\dots a_0 1$, and by Lucas's Theorem (Theorem~\ref{thm:lucas}), we have
	\[ \binom{2m-1}{m-1} \equiv \binom{1}{a_0}\cdot \prod_{i=1}^{r-1}\binom{a_{i-1}}{a_i}\cdot \binom{a_{r-1}}{0} \pmod{2}.\]
	This is non-zero if and only if $a_i=1$ for each $i$, or equivalently $m-1=2^{r}-1$.
\end{proof}

Finally, we prove our main theorem:
\begin{proof}[Proof of Theorem~\ref{thm:unbiased_even_nottwopow}]

Let $n\geq 1$ be an even integer that is not a power of two.
By Lemma~\ref{lemma:divisibility}, we have that $\binom{n}{n/2}$ is divisible by 4.
Now, the result follows by Lemma~\ref{lemma:comp_balanced_exist} and Lemma~\ref{lemma:comp_bal_iff_unbiased}.
\end{proof}
\section{Existence of Fair Voting Rules}

In this section, we prove our main results. 
We characterize for which $n$ a fair voting rule $f:2^N\to \set{0,1}$ exists, under different notions of fairness.
We begin by considering unbiased voting rules.

\subsection{Unbiased Voting Rules}

First, we show that unbiased voting rules exist for all $n$ except even powers of two.

\begin{theorem}
    Let $n\geq 1$ be an integer.
    An unbiased voting rule $f:2^N\to\set{0,1}$ exists if and only if $n \not \in \set{2,4,8,16,\dots}$. 
\end{theorem}

The proof of this theorem is carried out in two parts, as follows.

\begin{lemma}\label{lemma:unbiased_existence} (Existence of Unbiased Voting Rules)
    Let $n\geq 1$ be an integer.
    An unbiased voting rule $f:2^N\to\set{0,1}$ exists if $n \not \in \set{2,4,8,16,\dots}$. 
\end{lemma}
\begin{proof}
    Let $n\geq 1$ be an integer.
    For odd $n$, the majority voting rule on $n$ voters works.
    For even $n$, the result follows by Theorem~\ref{thm:unbiased_even_nottwopow}.
\end{proof}

\begin{lemma}\label{lemma:unbiased_nonexistence} (Non-Existence of Unbiased Voting Rules)
    Let $n\geq 1$ be an integer.
    An unbiased voting rule $f:2^N\to\set{0,1}$ exists only if $n \not \in \set{2,4,8,16,\dots}$. 
\end{lemma}
\begin{proof}
     Let $n\geq 2$ be an even integer such that there exists an unbiased voting rule on $n$ voters.
     By Lemma~\ref{lemma:comp_bal_iff_unbiased}, there exists a complementary balanced family $\mc T$ of $\frac{n}{2}$-sized subsets.
     In any such family, the number of subsets containing any fixed voter $i\in N$ must equal $\frac{1}{4}\binom{n}{n/2}$, and hence $\binom{n}{n/2}$ must be divisible by 4.
     By Lemma~\ref{lemma:divisibility}, this implies that $n \not \in \set{2,4,8,16,\dots}$.
\end{proof}

\subsection{Shapley-Shubik-Fair Voting Rules}

Second, we consider Shapley-Shubik-fair voting rules.
We show that---like for unbiasedness---such voting rules exist for all $n$ except even powers of two.

\begin{theorem}
    Let $n\geq 1$ be an integer.
    A Shapley-Shubik-fair voting rule $f:2^N\to\set{0,1}$ exists if and only if $n \not \in \set{2,4,8,16,\dots}$. 
\end{theorem}

Since any unbiased voting rule is also Shapley-Shubik fair, our result constructing unbiased voting rules already applies to yield the following first part of the proof.

\begin{lemma} (Existence of Shapley-Shubik-fair  Voting Rules)
    Let $n\geq 1$ be an integer.
    A Shapley-Shubik-fair voting rule $f:2^N\to\set{0,1}$ exists if $n \not \in \set{2,4,8,16,\dots}$. 
\end{lemma}
\begin{proof}
    This follows by Lemma~\ref{lemma:unbiased_existence} as unbiased voting rules are also Shapley-Shubik-fair (Proposition~\ref{prop:winning_coalition_properties}).
\end{proof}

The second part of the proof strengthens our impossibility result for unbiased voting rules.
We show that Shapley-Shubik-fair voting rules cannot exist when $n$ is an even power of 2.
Note that by Proposition~\ref{prop:winning_coalition_properties}, this gives an alternative proof of the non-existence of unbiased voting rules (Lemma~\ref{lemma:unbiased_nonexistence}).

\begin{lemma} (Non-Existence of Shapley-Shubik-fair Voting Rules)
\label{lem:shapley-non-existence}
    Let $n\geq 1$ be an integer.
    A Shapley-Shubik-fair voting rule $f:2^N\to\set{0,1}$ exists only if $n \not \in \set{2,4,8,16,\dots}$. 
\end{lemma}
\begin{proof}
     Let $n\geq 1$ be an integer such that there exists a Shapley-Shubik-fair voting rule on $n$ voters.
     Consider any voter $i\in N$. Since the Shapley-Shubik indices sum to 1, and $f$ is Shapley-Shubik-fair, it must hold that $\varphi_i = \frac{1}{n}$.
     By Proposition~\ref{prop:winning_coalition_properties}, this equals
     \[ \frac{1}{n} = \varphi_i = -1+2\cdot \sum_{k=1}^n \frac{ (k-1)! (n-k)!}{n!} \cdot w_i^{(k)} = -1+2\cdot \sum_{k=1}^n \frac{1}{n \cdot \binom{n-1}{k-1}} \cdot w_i^{(k)}, \]
     where $w_i^{(k)}$ denotes the number of winning coalitions of size $k$ including the voter $i$.
    Rearranging, we get
    \[ 2\cdot \sum_{k=1}^n \frac{w_i^{(k)}}{\binom{n-1}{k-1}} = n+1, \]
    and
    \[ 2\cdot \sum_{k=1}^n w_i^{(k)}\cdot \prod_{\substack{\ell=1\\ \ell\not=k}}^n\binom{n-1}{\ell-1}= (n+1) \cdot \prod_{k=1}^n\binom{n-1}{k-1}. \]
    This implies that 
    \[(n+1) \cdot \prod_{k=1}^n\binom{n-1}{k-1} \equiv 0 \pmod{2}.\]

    Suppose, for the sake of contradiction, that $n=2^r$ for some integer $r\geq 1$.
    Then, all the binomial coefficients $\binom{n-1}{k-1}$, for $k=1,2,\dots,n$, are odd; this follows by Lucas's Theorem (Theorem~\ref{thm:lucas}), as the binary representation of $n-1=2^r-1$ is $\underbrace{11\dots 1}_{r\text{ times }}$.
    This is a contradiction.   
\end{proof}

\subsection{Banzhaf-Fair Voting Rules}

Third, we consider Banzhaf-Fair voting rules.
In contrast to our results for Shapley-Shubik fairness, we show such rules exist for all $n$ except $2,4,$ and $8$.

\begin{theorem}
    Let $n\geq 1$ be an integer.
    A Banzhaf-fair voting rule $f:2^N\to\set{0,1}$ exists if and only if $n \not \in \set{2,4,8}$. 
\end{theorem}

The proof of the above theorem is carried out in two parts, as follows:

\begin{lemma} (Existence of Banzhaf-fair Voting Rules)
\label{lem:exists-banzhaf-fair}
    Let $n\geq 1$ be an integer.
    A Banzhaf-Fair voting rule $f:2^N\to\set{0,1}$ exists if $n \not \in \set{2,4,8}$.
\end{lemma}
\begin{proof}
    For $n \not \in \set{2,4,8,16,\dots}$, this follows by Lemma~\ref{lemma:unbiased_existence} as unbiased voting rules are also Banzhaf-fair (Proposition~\ref{prop:winning_coalition_properties}).

    For $n\in \set{16, 32, 64, \dots}$, the result follows by by a construction of regular maximal intersecting families, given in Theorem 5 of~\citet{Mey95}.\footnote{
      Maximal intersecting families---i.e., maximal families $\mc F$ of subsets of $\set{1,2,\dots,n}$ that pairwise intersect---have received considerable attention in the combinatorics literature.
      One can show that (1) all such $\mc F$ are upwards-closed, and (2) all such $\mc F$ contain exactly one of $S$ or $N\setminus S$ for each $S\subseteq N$ (see, e.g., \citealt[Lemma 2.1, Corollary 2.2]{Mey95}). 
      Hence, such families yield the set of winning coalitions of a monotone, neutral, and resolute voting rule.
      Regular families are those such that each $i\in N$ is contained in an equal number of elements of $\mc F$.
      Hence, by Proposition~\ref{prop:winning_coalition_properties}, regular maximal intersecting families yield Banzhaf-fair voting rules.
    }
\end{proof}

\subsubsection{Non-Existence of Banzhaf-fair Voting Rules}

We prove that no Banzhaf-Fair voting rule $f:2^N\to\set{0,1}$ exists for $n=2,4,8$.
The proof proceeds via case analysis.

\begin{lemma}\label{lemma:banzhaf_two}
    There is no Banzhaf-fair voting rule on $n=2$ voters.
\end{lemma}
\begin{proof}
    Let $f:2^N\to \set{0,1}$ be any voting rule with $n=2$.
    We know that \[ f(\emptyset)=0,\quad f (\set{1,2})=1, \quad f(\set{1}) = 1-f(\set{2}).\]
    Thus, the only possible voting rules are dictators, which are easily verified to not be Banzhaf-fair.
\end{proof}

\begin{lemma}\label{lemma:banzhaf_four}
    There is no Banzhaf-fair voting rule on $n=4$ voters.
\end{lemma}
\begin{proof}
    Let $f:2^N\to \set{0,1}$ be any voting rule with $n=4$.
    We know $f(\emptyset)=0,\ f(N)=1$.
    
    Suppose there exists $S\subseteq N, |S|=1$ with $f(S)=1$; without loss of generality $S=\set{1}$. Then, since $f$ is monotone and neutral, it must equal the dictator function with respect to voter 1, and this is easily verified to not be Banzhaf-fair.

    Suppose $f(S)=0$ for each $S\subseteq N, |S|=1$.
    Then, by neutrality, $f(S)=1$ for each $S\subseteq N, |S|=3$.
    By Proposition~\ref{prop:winning_coalition_properties}, any such function that is Banzhaf-fair must also be unbiased. 
    Since there exists no unbiased voting rule on $n=4$ voters (Lemma~\ref{lemma:unbiased_nonexistence}), there is no such function which is Banzhaf-fair.
\end{proof}

\begin{lemma}\label{lemma:banzhaf_eight}
    There is no Banzhaf-fair voting rule on $n=8$ voters.
\end{lemma}

\begin{proof}
Let $f:2^N\to \set{0,1}$ be a voting rule on $n=8$ voters.
We know $f(\emptyset) = 0$ and $f(N)=1$.
Suppose, for the sake of contradiction, that $f$ is Banzhaf-fair.
Consider the following cases:
\paragraph{Winning Coalition of Size 1:}  Suppose there exists $S\subseteq N, |S|=1$ with $f(S)=1$; without loss of generality $S=\set{1}$. Then, since $f$ is monotone and neutral, it must equal the dictator function with respect to voter 1, and this is easily verified to not be Banzhaf-fair.

\paragraph{Winning Coalition of Size 2:}  Now, suppose there exists no winning coalition of size 1, and that there exists $S\subseteq N, |S|=2$ with $f(S)=1$; without loss of generality $S=\set{1,2}$.

For each $S\subseteq \set{1,2}$, define $w_S$ to be the number of winning coalitions containing exactly voters in $S$ among the first two voters; i.e.,
\[ w_S = \abs{\set{T\subseteq \set{3,4,\dots,8} : f(S\cup T) =1}}. \]
For simplicity of notation, we use $w_1,w_2,w_{12}$ to denote $w_{\set{1}}, w_{\set{2}}, w_{\set{1,2}}$ respectively.
Note that $w_{12} = 64$ by monotonicity, and $w_{\emptyset}=0$ as winning coalitions are intersecting (Lemma~\ref{lemma:wc_intersect}).

Observe that for any $T\subseteq \set{3,4,\dots,8}$, we have $N\setminus (T\cup\set{1}) = (\set{3,4,\dots,8}\setminus T)\cup\set{2}$.
Since neutrality of $f$ implies that exactly one of $U$ or $N\setminus U$ is winning for each $U$, we have $w_1+w_2 = 64$.
Now, the total number of winning coalitions including voter 1 (resp. voter 2) equals $w_{12}+w_1$ (resp. $w_{12}+w_2$).
Since $f$ is Banzhaf-fair, these must be equal, and hence $w_1=w_2=32$.
In particular, this also implies that each voter $i\in N$ is a part of $w_{12}+w_1=64+32=96$ winning coalitions.

Now, consider all winning coalitions containing voter 3.
Since all winning coalitions are intersecting (Lemma~\ref{lemma:wc_intersect}), every such coalition is of the form $S\cup \set{3}\cup T$, where $\emptyset\not= S\subseteq \set{1,2}$ and $T\subseteq \set{4,5,\dots,8}$.
The total number of such coalitions is $96=3\cdot 32$, and since voter 3 occurs in 96 winning coalitions, all of these must be winning.
In particular, we have $f(\set{1,3})=f(\set{2,3})=1$.

A similar reasoning for voter 4 shows that $f(\set{1,4})=f(\set{2,4})=1$.
This is a contradiction since $\set{1,3}$ and $\set{2,4}$ are disjoint.

\paragraph{Winning Coalition of Size 3:} Now, suppose there exists no winning coalition of sizes 1,2, and that there exists $S\subseteq N, |S|=3$ with $f(S)=1$; without loss of generality $S=\set{1,2,3}$.

For each $S\subseteq \set{1,2,3}$, define $\mc W_S$ to be the set of voters in $\set{4,5,\dots,8}$ who form a winning coalition alongside $S$,
and let $w_S$ be the number of such coalitions; i.e.,
\[\mc W_S = \set{T\subseteq \set{4,5,\dots,8} : f(S\cup T) =1}, \quad w_S=\abs{\mc W_S}.\]
For simplicity of notation, we use $w_1$ to denote $w_{\set{1}}$, and $w_{123}$ to denote $w_{\set{1,2,3}}$, and so on.
Note that $w_{123} = 32$ by monotonicity, and $w_{\emptyset}=0$ as winning coalitions are intersecting (Lemma~\ref{lemma:wc_intersect}).
Now, observe the following:
\begin{enumerate}
	\item By the neutrality of $f$, we have the following relations \[ w_1+w_{23} = w_2+w_{13} = w_3+w_{12}=32. \]
	For example, $w_1+w_{23}=32$ follows by observing that for any $T\subseteq \set{4,5,\dots,8}$, we have $N\setminus (T\cup\set{1}) = (\set{4,5,\dots,8}\setminus T)\cup\set{2,3}$.
	
	\item Since $f$ is Banzhaf-fair, we also have the following relations \[ w_1+w_{12}+w_{13} = w_2+w_{12}+w_{23}=w_3+w_{13}+w_{23}. \]
	This follows by observing that the number of winning coalitions containing voter 1 is $w_1+w_{12}+w_{13}+w_{123}$ (and similarly for other voters).
\end{enumerate}

Solving the above system of linear equations, we get
\[
	w_1=w_2=w_3=s, \quad w_{12}=w_{13}=w_{23} = 32-s, \quad \text{for some integer }0\leq s\leq 16,
\]
where we used that by monotonicity, $w_1\leq w_{12}$, and hence $s\leq 16$.
Moreover, each voter $i\in N$ occurs in $w_1+w_{12}+w_{13}+w_{123} = 96-s$ winning coalitions.

Next, we measures the sum of the sizes of all winning coalitions.
Since each voter occurs in $96-s$ winning coalitions, this equals
\begin{equation}\label{eqn:coal_count_1}
	\sum_{S\subseteq N: f(S)=1}|S| = 8\cdot(96-s) = 768-8s.
\end{equation}
On the other hand, we count this sum by summing over different set sizes $|S|$.
Define $t$ to be the number of winning coalitions of size 3, i.e., $t= \abs{\set{S\subseteq N: f(S)=1, |S|=3}}$. Then,
\begin{enumerate}
	\item We know $f(S)=0$ for $|S|=0,1,2$, and hence, the contribution from such sets is 0.
	\item By neutrality, we have that $f(S)=1$ whenever $|S|=6,7,8$. The contribution from such sets is $6\cdot\binom{8}{6} + 7\cdot\binom{8}{7}+ 8\cdot\binom{8}{8} = 232$.
	\item By neutrality, $f(S)=1$ for exactly $\frac{1}{2}\binom{8}{4}=35$ sets of size $|S|=4$. The contribution from such sets is $4\cdot 35 = 140$.
	\item The contribution from sets of size 3 is $3t$.
	\item By neutrality, the number of losing coalitions of size 5 is $t$, and hence the number of winning coalitions of size 5 is $\binom{8}{5}-t = 56-t$. The contribution from such sets is $5\cdot (56-t) = 280-t$.
\end{enumerate}
Summing all of the above, we get that 
\[ \sum_{S\subseteq N: f(S)=1}|S| = 0+232+140+3t+280-5t=652-2t. \]
Combining with Equation~\eqref{eqn:coal_count_1}, we get $768-8s = 652-2t$, or equivalently, $t = 4s-58.$
Since $t\geq 0$, we get $s\geq 15$.
Recall that $s\leq 16$; hence, the only possible combinations of $s,t$ are $s=15,\ t=2$ and $s=16,\ t=6$.
We consider these two cases one by one.

\begin{itemize}
	\item Suppose $s=16,\ t=6$.
		In this case, $w_1=w_{12}=16$, and since $\mc W_1\subseteq \mc W_{12}$, we have $\mc W_1=\mc W_{12}$. Using other such equalities, we get \[\mc W_1=\mc W_2=\mc W_3=\mc W_{12}=\mc W_{23}=\mc W_{13} := \mc U \subseteq \set{4,5,\dots,8}.\]
		Now, consider voter 4, and let $k = \abs{\set{T\in \mc U: 4\in T}}$ be the number of sets in $\mc U$ containing voter 4.
		The total number of winning coalitions containing voter 4 is
		\[ \sum_{S\subseteq \set{1,2,3}}\abs{\set{T\subseteq \set{4,5,\dots,8}: 4\in T, f(S\cup T)=1}} = 6k+16,\]
		because the sum is $0$ for $\abs{S}=0$, and $k$ for $|S|=1$ or $2$, and 16 for $|S|=3$.
		Since this equals $96-s=80$, we have $6k+16=80$, and $6k=64$, which is not possible for an integer $k$. This is a contradiction.
	\item Suppose $s=15,\ t=2$. 	
	Consider the two winning coalitions of size 3, one of them is $\set{1,2,3}$, and let the other one be $A$. Since winning coalitions are intersecting, we may assume without loss of generality that $A = \set{1,2,4}$ or $A = \set{1,4,5}$.
		We do a case analysis:
	
	\begin{itemize}
		\item Suppose $A = \set{1,2,4}$. Consider the set $\mc W_{12}$, of size $32-s=17$. Observe that $\set{4}\in \mc W_{12}$, and hence by monotonicity, every $T\subseteq \set{4,5,\dots,8}$ with $4\in T$ is such that $T\in \mc W_{12}$.
			Now, $\mc W_1\subseteq \mc W_{12}$ is of size $s=15$, and hence excludes exactly 2 elements from $\mc W_{12}$.
			Since $\set{4}, \set{4,5}, \set{4,6} \in \mc W_{12}$, at least one of them must lie in $\mc W_1$.
			Thus, at least one of $\set{1,4}, \set{1,4,5}, \set{1,4,6}$ is a winning coalition, which is a contradiction to the assumption that the only winning coalitions of size at most 3 are $\set{1,2,3}$ and $\set{1,2,4}$.		
		\item Suppose $A = \set{1,4,5}$.			
		Consider the set $\mc W_2$ of size $s=15$. This cannot have any set of size $0,1$, or $2$, as that would lead to a new winning coalition of size at most $3$. Also, $\set{6,7,8}\not\in \mc W_2$ or else $\set{2,6,7,8}$ is a winning coalition disjoint from $\set{1,4,5}$. We remain with exactly 15 subsets of $\set{4,5,\dots,8}$, and hence \[\mc W_2 = \set{T\subseteq \set{4,5,\dots,8}: |T|\geq 3} \setminus \set{6,7,8}.\]
		Similarly, considering voter 3, we also get	
		\[\mc W_3 = \set{T\subseteq \set{4,5,\dots,8}: |T|\geq 3} \setminus \set{6,7,8}.\]
	
		Now, we count the number of winning coalitions containing voter 4.
		For each $S\subseteq \set{1,2,3}$, define
		\[ \mc V_{S} = \set{T\in \mc W_S: 4\in T},\quad v_S=|\mc V_S|. \]
		Observe that $v_{\emptyset}=0,\ v_{123}=16$.
		Also, for any set $S\subseteq \set{1,2,3}$ of size $|S|=1$, and $S' = \set{1,2,3}\setminus S$, we have by neutrality that
		\[ v_S = w_{S}-(2^4-v_{S'}) = 15-2^4+v_{S'}= v_{S'}-1. \]
		Now, since we know exactly what $\mc W_2,\mc W_3$ are, direct counting shows $v_2=v_3=11$, and hence $v_{12}=v_{13}=12$.
		Also, $v_{23}=v_1+1$, and so counting the total number of winning coalitions containing voter 4, we get
		\[ 81=96-s = 0 + v_1+ 11 + 11 + 12+12+(v_1+1)+16 = 2v_1+63, \]
		and hence $v_1=9$ and $v_{23}=10$.
		But, by monotonicity, $11=v_2\leq v_{23} = 10$, which is a contradiction.
	\end{itemize}
\end{itemize}

\paragraph{No Winning Coalition of Size 1,2,3:} 
Suppose that $f(S)=0$ for every $S\subseteq N$ with $|S|\leq 3$.
By neutrality, $f(S)=1$ for every $S\subseteq N$ with $|S|\geq 5$.
By Proposition~\ref{prop:winning_coalition_properties}, since $f$ is Banzhaf-fair, it must also be unbiased. 
Since there exists no unbiased voting rule on $n=8$ voters (Lemma~\ref{lemma:unbiased_nonexistence}), we get a contradiction.
\end{proof}

{
\printbibliography
}

\appendix

\section{Lucas's Theorem}

\begin{theorem}\label{thm:lucas}
    (Lucas's Theorem modulo 2)
    Let $m,n\geq 1,\ n\leq m$ be integers such that their binary representations are given by
    \[ m = a_r 2^r + a_{r-1}2^{r-1}+\dots+a_12 + a_0,\]
    \[ n = b_r 2^r + b_{r-1}2^{r-1}+\dots+b_12 + b_0.\]
    Then, 
    \[ \binom{m}{n} \equiv \prod_{i=0}^r \binom{a_i}{b_i} \pmod{2}. \]
\end{theorem}


\section{Example: An Unbiased but not Equitable Voting Rule}
\label{sec:size-9-unbiased}

We now provide an example of an unbiased (and hence Shapley-Shubik-fair and Banzhaf-fair) voting rule that is not equitable, on $n=9$ voters.

We define the rule in terms of its minimal winning coalitions (MWCs) of size $4$, given by:
\begin{align*}
    S_1&=\set{1,2,3,4}, & S_2&=\set{4,5,6,9}, & S_3&=\set{3,7,8,9}, \\
    S_4&=\set{1,4,5,7}, & S_5&=\set{2,5,6,8}, & S_6&=\set{3,6,7,9}, \\
    S_7&=\set{1,2,5,9}, & S_8&=\set{2,6,7,8}, & S_9&=\set{1,3,4,8}.
\end{align*}
Define $\mc W^{(4)} = \set{S_i: i=1,2,\dots, 9}$ to be the collection of all such sets.
Now, let $f$ be the function defined as follows:
\[ f(S)=\begin{cases}
    1, & W\subseteq S \text{ for some } W\in \mc W^{(4)} \\
    0, & W\subseteq N\setminus S \text{ for some } W\in \mc W^{(4)} \\
    \ind\bigl[{|S|\geq 5}\bigr], & \text{otherwise.}
\end{cases}.
\tag{($\dagger$)}
\label{eqn:middle-leval-plus-tie}\]
In words, $f$ outputs $1$ or $0$ if any of the sets in $\mc W^{(4)}$ vote for $1$ or $0$; otherwise, $f$ ``breaks the tie'' by taking a majority of the votes.\footnote{A similar approach to defining voting rules given a pairwise-intersecting family of sets is adopted in \citet{BartholdiHJTY21}.}\textsuperscript{,}\footnote{Another way to see the same function is the following: for sets of size at most 3, $f$ is 0; for sets of size 4, $f$ is 1 for exactly the sets in $\mc W^{(4)}$; for sets of size at least $5$, $f$ is uniquely determined by the neutrality property.}

\begin{proposition}
  The function $f$ is a well-defined voting rule which is unbiased, but is not equitable.
\end{proposition}
The remainder of this section is devoted to the proof of the above proposition.
First, we prove some crucial properties about the set $\mc W^{(4)}$.

\begin{lemma}\label{lemma:pair_int}
    \begin{enumerate}
        \item The collection $\mc W^{(4)}$ is pairwise intersecting, i.e., for each $S,T\in \mc W^{(4)}$, it holds that $S\cap T\not=\emptyset$.
        \item Each voter $i\in N$ occurs in 4 sets in $\mc W^{(4)}$.
    \end{enumerate}
\end{lemma}
\begin{proof}
    Both the items are easily verified.
\end{proof}

Next, we show that the function defined by \autoref{eqn:middle-leval-plus-tie} is a valid voting rule:

\begin{lemma}
    The function $f$ is a well-defined voting rule.\footnote{More generally our proof shows the following: Suppose $n=2r+1$ is an odd number of voters, and suppose $\mc W^{(r)}$ is a pairwise-intersecting family of size-$r$ subsets of $\set{1,2,\dots,n}$.
    Then the rule as in \autoref{eqn:middle-leval-plus-tie} (where $\mc W^{(r)}$ are the MWCs of size $r$, and ties are broken by taking majority) is a well-defined voting rule.}
\end{lemma}
\begin{proof}
    The proof is carried out in three steps:

    \paragraph{Well-defined:} Note that \autoref{eqn:middle-leval-plus-tie} defines a valid function if and only if the first two conditions never occur simultaneously.
    That is, we wish to show that there is no $S\subseteq N$ and $ W,W'\in \mc W^{(4)}$ satisfying $W\subseteq S$ and $W'\subseteq N\setminus S$.
    This follows since $\mc W^{(4)}$ is pairwise intersecting (Lemma~\ref{lemma:pair_int}).

    \paragraph{Monotone:} Let $S\subseteq T \subseteq N$ be arbitrary sets.
    If $f(S)=0$, then it holds that $f(S)\leq f(T)$ and there is nothing to prove.
    Now, suppose that $f(S)=1$.
    By \autoref{eqn:middle-leval-plus-tie}, there are two possibilities:
    \begin{itemize}
        \item There exists $W\in \mc W^{(4)}$ with $W\subseteq S$. Then, it also holds $W\subseteq T$, and hence $f(T)=1$.
        \item $|S|\geq 5$, and for every $W\in \mc W^{(4)}$, $W\not\subseteq (N\setminus S)$.
        Then, we have $|T|\geq |S|\geq 5$, and also for every $W\in \mc W^{(4)}$, $W\not\subseteq (N\setminus T)\subseteq  (N\setminus S)$. Hence, $f(T)=1$.
    \end{itemize}

    \paragraph{Neutral:} For all $S\subseteq N$, it is easily verified for each case in \autoref{eqn:middle-leval-plus-tie} that $f(N\setminus S) = 1-f(S)$.
\end{proof}

Next, we show that $f$ is unbiased.

\begin{proposition}
    The voting rule $f$ is unbiased.\footnote{More generally our proof shows the following: Suppose $n=2r+1$ is an odd number of voters. Let $\mc W^{(r)}$ be a pairwise-intersecting family of size-$r$ subsets of $\set{1,2,\dots,n}$, and 
    suppose that all voters $i\in N$ are contained in the same number of sets from $\mc W^{(r)}$.
    Then, the voting rule constructed as in \autoref{eqn:middle-leval-plus-tie} (where $\mc W^{(r)}$ are the MWCs of size $r$, and ties are broken by taking majority), is unbiased.}
\end{proposition}
\begin{proof}
    For each $k=0,1,2,\dots,n$, and $i\in N$, let $w_i^{(k)}$ denote the number of winning coalitions of size $k$, i.e., \[w_i^{(k)} = \abs{\set{S\subseteq N :  f(S)=1,\ S\ni i,\ \abs{S}=k}}.\]
    By Proposition~\ref{prop:unbiased_win_coal}, it suffices to prove that for each $i,j\in N$, and each $k\in \set{0,1,2,\dots,n}$, it holds that $w_i^{(k)} = w_j^{(k)}$.
    By neutrality, it suffices to show that this holds for all $0\leq k\leq 4$, because for any $4 < k \leq 9$,  we have
    \[w_i^{(k)} =  w_i^{(n-k)} + \binom{n-1}{n-k} - w^{(n-k)}, \]
    where $w^{(n-k)}$ denotes the total number of winning coalitions of size $n-k$.
    
    First, observe that by \autoref{eqn:middle-leval-plus-tie}, we have $f(S)=0$ for each $S\subseteq N$ with $|S|\leq 3$.
    Hence, for all $k\leq 3$, and $i,j\in N$, we have $w_i^{(k)}=w_j^{(k)}$.

    Now, suppose $k=4$.
    Then, by \autoref{eqn:middle-leval-plus-tie}, we observe that the set of winning coalitions of size $4$ exactly equals $\mc W^{(4)}$.
    Since each voter is contained in the same number of sets in $\mc W^{(4)}$, we have $w_i^{(4)}=w_j^{(4)}$ for all $i,j\in N$.
\end{proof}

Finally, we show that $f$ is not equitable.

\begin{lemma}
    The voting rule $f$ is not equitable.
\end{lemma}
\begin{proof}
    We show that that there is no symmetry $\sigma:N\to N$ of $f$ such that $\sigma(1)=2$.
    Suppose, for the sake of contradiction, that $\sigma$ is such a symmetry.
    Observe that $\sigma$ is also a symmetry of $\mc W^{(4)}$, i.e., for each $S\in \mc W^{(4)}$, it holds that $\sigma(S)\in \mc W^{(4)}$.
    This is because by \autoref{eqn:middle-leval-plus-tie}, the collection $\mc W^{(4)}$ is exactly the set of all size-4 winning coalitions of $f$.

    Let $\mc W_1^{(4)}$ (resp. $\mc W_2^{(4)}$) be the collection the sets in $\mc W^{(4)}$ including voter 1 (resp. voter 2), i.e.,
    \[ \mc W_1^{(4)} = \set{ \set{1,2,3,4}, \set{1,4,5,7}, \set{1,2,5,9},\set{1,3,4,8} }, \]
    \[ \mc W_2^{(4)} = \set{\set{1,2,3,4}, \set{2,5,6,8}, \set{1,2,5,9}, \set{2,6,7,8}}. \]
    Since $\sigma(1)=2$, and $\sigma$ is a symmetry of $\mc W^{(4)}$, it must hold that $\sigma$ maps $ \mc W_1^{(4)}$ to $\mc W_2^{(4)}$.
    Hence, the multi-set of pairwise intersection sizes must be the same for $\mc W_1^{(4)}$ and $\mc W_2^{(4)}$, i.e.,
    \[ \set{|S\cap T|: S\not=T\in \mc W_1^{(4)}} = \set{|S\cap T|: S\not=T\in \mc W_2^{(4)}}. \]
    However, direct calculation shows that the multi-set on the left equals $\set{1,2,2,2,2,3}$, while the multi-set on the right equals $\set{1,1,1,2,2,3}$, which is a contradiction.
\end{proof}


\section{Example: An Equitable Voting Rule}\label{app:equit_ex}

We now present an additional example of an equitable voting rule, which may be of independent interest.
We define the rule in terms of its minimal winning coalitions (MWCs) of size $4$.
These MWCs are constructed from the triangular prism graph in \autoref{fig:size-9-unbiased} as follows.\footnote{
  We found this voting rule when trying to construct a non-equitable (but unbiased) rule in terms of the non-transitive-symmetric (but degree-regular) line graph of \autoref{fig:size-9-unbiased}. To our surprise, however, the corresponding voting rule turned out to be equitable.
}

\begin{figure}
    \centering
    \includegraphics[width=0.5\linewidth]{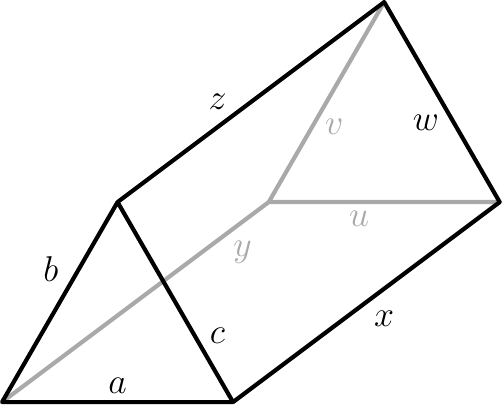}
    \caption{A triangular prism}
    \Description{A Triangular Prism.}
    \label{fig:size-9-unbiased}
\end{figure}

Let $E = \set{a,b,c,x,y,z,u,v,w}$ denote the edges of the graph, which we identify with the set of 9 voters.
For each edge $e\in E$, let $N(e)$ denote all the edges adjacent to edge $e$.
These are:
\begin{align*}
    N(a)&=\set{b,c,x,y}, & N(b)&=\set{a,c,y,z}, & N(c)&=\set{a,b,x,z}, \\
    N(x)&=\set{a,c,u,w}, & N(y)&=\set{a,b,u,v}, & N(z)&=\set{b,c,v,w}, \\
    N(u)&=\set{x,y,v,w}, & N(v)&=\set{y,z,u,w}, & N(w)&=\set{x,z,u,v}. 
\end{align*}

Define $\mc W^{(4)} = \set{N(e): e\in E}$ to be the collection of all such sets.
This is easily verified to be pairwise-intersecting.
Now, let $f$ be the voting rule defined as in \autoref{eqn:middle-leval-plus-tie}.

\begin{lemma}
    The voting rule $f$ is equitable.
\end{lemma}
\begin{proof}[Proof Sketch]
    The permutation $\sigma=(a,x)(b,z)(c,y)(u,v)(w)$, represented in cycle notation, is a symmetry of $\mc W^{(4)}$ mapping $a$ to $x$.
    Note that defines a permutation on $\mc W^{(4)}$, given by \[(N(a), N(b))\ (N(c))\ (N(x), N(u))\ (N(y), N(w))\ (N(z), N(v)).\]

    For any other pair of edges $e,e' \in E$, we can compose the above symmetry with rotations and reflections (on the prism graph) to get a symmetry mapping $e$ to $e'$.
\end{proof}


\end{document}